\documentclass[envcountsame,a4paper]{llncs}
\usepackage{ifthen}
\newcommand{\techRep}{false} 
\newcommand{\iftechrep}{\ifthenelse{\equal{\techRep}{true}}}

\usepackage{a4wide}
\usepackage{amsmath}
\usepackage{amssymb}
\usepackage{btran}
\usepackage{tikz}
\usetikzlibrary{calc,automata}

\newcommand{\N}{\mathbb{N}}

\newcommand{\q}[1]{\overrightarrow{\langle\textup{\textsf{#1}}\rangle}}
\newcommand{\qr}[1]{\langle \textup{\textsf{#1}} \rangle_\bullet}
\newcommand{\ql}[1]{{}_\bullet\langle \textup{\textsf{#1}} \rangle}
\newcommand{\Att}{\mathit{Att}}
\newcommand{\Def}{\mathit{Def}}

\newcommand{\Tow}{\mathsf{Tower}}
\newcommand{\Act}{\mathcal{A}\mathit{ct}}

\renewcommand{\P}{\mathcal{P}}
\renewcommand{\S}{\mathcal{S}}

\newcommand{\val}{\mathit{val}}
\newcommand{\T}{\mathcal{T}}
\newcommand{\Bis}[1]
{\mathord{\sim}#1}
\newcommand{\shuf}{|\!|}
\newcommand{\sta}{$\beta$}

\newcommand{\defeq}{\stackrel{\text{def}}{=}}

\newcommand{\problemx}[3]{
\par\noindent\underline{\sc#1}\par\nobreak\vskip.2\baselineskip
\begingroup\clubpenalty10000\widowpenalty10000
\setbox0\hbox{\bf INPUT:\ \ }\setbox1\hbox{\bf QUESTION:\ \ }
\dimen0=\wd0\ifnum\wd1>\dimen0\dimen0=\wd1\fi
\vskip-\parskip\noindent
\hbox to\dimen0{\box0\hfil}\hangindent\dimen0\hangafter1\ignorespaces#2\par
\vskip-\parskip\noindent
\hbox to\dimen0{\box1\hfil}\hangindent\dimen0\hangafter1\ignorespaces#3\par
\endgroup}

\newcommand{\problemy}[3]{
\par\noindent\underline{\sc#1}\par\nobreak\vskip.2\baselineskip
\begingroup\clubpenalty10000\widowpenalty10000
\setbox0\hbox{\bf INPUT: }\setbox1\hbox{\bf OUTPUT: }
\dimen0=\wd0\ifnum\wd1>\dimen0\dimen0=\wd1\fi
\vskip-\parskip\noindent
\hbox to\dimen0{\box0\hfil}\hangindent\dimen0\hangafter1\ignorespaces#2\par
\vskip-\parskip\noindent
\hbox to\dimen0{\box1\hfil}\hangindent\dimen0\hangafter1\ignorespaces#3\par
\endgroup}

\title{Bisimilarity of Pushdown Systems is Nonelementary}
\author{Michael Benedikt\inst{1} \and Stefan G\"oller\inst{2} \and Stefan Kiefer\inst{1} \and Andrzej S.~Murawski\inst{3}}
\institute{University of Oxford, UK\\
\and LIAFA Paris, France and Universit\"at Bremen, Germany
\thanks{The research leading to these results has received funding from the European Union's Seventh
Framework Programme (FP7/2007-2013) under grant agreement n$^\circ$ 259454.}\\
\and University of Leicester, UK}
\begin{document}

\maketitle

\begin{abstract}
 Given two pushdown systems, the bisimilarity problem asks whether they are  bisimilar.
While this problem is known to be decidable our main result states that it is nonelementary, 
improving $\mathsf{EXPTIME}$-hardness, which was the previously best known lower bound for this problem.
Our lower bound result holds for normed pushdown systems as well.
\end{abstract}

\section{Introduction} \label{sec-introduction}

A central problem in theoretical computer science is to decide
whether two machines or systems behave equivalently.
While being generally undecidable for Turing machines,
a lot of research has been devoted to find subclassses of machine devices
for which this problem becomes decidable.
{\em Equivalence checking} is the problem of determining whether two systems are semantically identical.

It is well-known that even language equivalence of pushdown automata
is undecidable, in fact already their universality is undecidable.
On the positive side, a celebrated result due to S\'{e}nizergues states that 
language equivalence of deterministic pushdown automata is decidable \cite{Senizergues01}.
The best known upper bound for the latter problem is a tower of exponentials
 \cite{Stir02} (see \cite{Jan12} for a more recent proof), while only
hardness of deterministic polynomial time is known to date.

Among the numerous notions of equivalence \cite{vGlab90} 
in the realm of formal verification and concurrency theory,
the central one is 
{\em bisimulation equivalence} ({\em bisimilarity} for short),
which enjoys pleasant mathematical properties.
It can be seen to take the king role:
There are important characterizations the bisimulation-invariant fragments
of first-order logic and of monadic second-order logic
in terms of modal logic \cite{vBen76} and of the modal $\mu$-calculus \cite{JaWa96}, respectively.
In particular, bisimilarity is a fundamental notion for process algebraic formalisms~\cite{Milner}.
As a result, a great deal of research in the analysis of infinite-state systems 
 (such as pushdown systems or Petri nets)
 has been devoted to deciding bisimilarity of two given processes,
  see e.g.~\cite{KuceraJ06} for a comprehensive overview.

A milestone result in this context has been proven by S\'{e}nizergues:
Bisimilarity on pushdown systems (i.e. transition systems induced by pushdown automata without $\varepsilon$-transitions)
 is decidable \cite{Senizergues05} 
; in fact, in \cite{Senizergues05} bisimilarity is proven to be decidable for the
more general class of equational graphs of finite out-degree.
Since pushdown systems can be viewed as an abstraction of the call-and-return behavior of a
recursive program, the latter decidability result should be read as that one can decide equivalence
of recursive programs in terms of their visible behavior.
Concerning {\em decidability} the latter result can in some sense be considered as best possible
since on the slightly more general classes of 
type -1a and type -1b rewrite systems \cite{JS08} and order-two pushdown graphs \cite{BG12}
bisimilarity becomes undecidable.

Though being decidable, S{\'e}nizergues' algorithm for deciding bisimilarity of pushdown systems consists of
two semi-decision procedures and in fact no complexity-theoretic upper bound
is known for this problem to date.
On the other hand, the best known lower bound for this problem is $\mathsf{EXPTIME}$ 
 shown by Ku\v{c}era and Mayr \cite{KM10}.
In \cite{Kie12} $\mathsf{EXPTIME}$-hardness has been even be established even for
the subclass basic process algebras,
for which a $2\mathsf{EXPTIME}$ upper bound is known \cite{Burkart95}.
Such complexity gaps are typical in the context of infinite-state systems.

In fact, in case decidability is known, 
the precise {\em computational complexity} status of bisimilarity on
infinite-state systems is known only for few classes, let us mention
basic parallel processes (communication-free Petri nets) \cite{Jan03} and
 one-counter systems (the transition systems induced by pushdown automata
over a singleton stack alphabet) \cite{BGJ10}.

{\bf Our contribution.} 
The main result of this paper states that bisimilarity of
(systems induced by) pushdown systems is nonelementary, even in the normed case.
We give small descriptions of pushdown systems on which a bisimulation game is
implemented that
allows to push and verify encodings of nonelementarily big counters \`{a} la Stockmeyer
\cite{Sto74}.
 As an important technical tool we realize deterministic verification phases in the bisimulation
game by simulating non-erasing real-time transducers that are fed with the stack content.
As basic gadgets, we use the well-established technique of Defender's Forcing \cite{JS08}.
We are optimistic that our technique gives new insights for potential
further lower bounds for
bisimilarity of PA processes, regularity for pushdown systems, and
weak bisimilarity of basic process algebras.

{\bf Organisation.} 

In Section \ref{sec-preliminaries} we introduce preliminaries.
In Section \ref{sec-techniques} we recall basics on transductions, introduce useful abbreviations for 
pushdown rules and recall Defender's forcing.
Section \ref{sec-construction} consists of our nonelementary lower bound proof for bisimilarity
of pushdown systems.


\section{Preliminaries} \label{sec-preliminaries}
By $\N\defeq\{0,1,\ldots\}$ we denote the set of {\em non-negative integers}.
For $n,m \in \N$ with we write $[n,m]$ for $\{n, n+1, \ldots, m\}$;
in particular note that $[n,m]=\emptyset$ if $n>m$.

A {\em labelled transition system (LTS)}  is a tuple
$\S=(S,\Act,\{\mathord{\xrightarrow{a}}\mid a\in\Act\})$,
where $S$ is a set of {\em configurations},
$\Act$ is a finite set of {\em action labels},
and $\mathord{\xrightarrow{a}}\subseteq S\times S$ is
a {\em transition relation} for each $a\in\Act$.
We say that a state $s\in S$ is a {\em deadlock} if there is no $t\in S$
and no $a\in\Act$
such that $s\xrightarrow{a}t$.
A binary relation $R\subseteq S\times S$ is a {\em bisimulation}
if for each $(s,s')\in R$ and each $a\in\Act$,
we have: (1) if $s\xrightarrow{a}t$, then there is some
$s'\xrightarrow{a}t'$ with $(t,t')\in R$
and, conversely,
(2) if $s'\xrightarrow{a}t'$, then there is some
$s\xrightarrow{a}t$ with $(t,t')\in R$.
We write $s\sim t$ is there is some bisimulation
$R$ with $(s,t)\in R$. 
Although not explicitly used in this paper, it is sometimes convenient to view bisimilarity
as a game between Attacker and Defender. 
In every round of the game, there is a pebble placed on a unique state in
each transition system. Attacker then chooses one
transition system and moves the pebble from the pebbled state to one
of its successors by an action $\xrightarrow{a}$, where $a$ is some action label.
Defender must imitate this by moving the pebbled state from the other
system to one of its successors by the same action $\xrightarrow{a}$.
 If one player cannot move, then the other player wins. 
Defender wins every infinite game. 
Two states $s$ and $t$ are bisimilar if and only if Defender
has a winning strategy on the game with initial
pebble configuration $(s,t)$.



A \emph{pushdown automaton (PDA)} is a tuple $\P=(Q,\Gamma,\Act,\mathord{\btran{}})$,
 where $Q$ is a finite set of \emph{control states},
 $\Gamma$ is a finite set of \emph{stack symbols},
 $\Act$ is a finite set of \emph{actions},
 and $\mathord{\btran{}} \subseteq
  (Q \times \{\varepsilon\} \times \Act \times Q \times \{\varepsilon\}) \cup
  (Q \times \{\varepsilon\} \times \Act \times Q \times \Gamma) \cup
  (Q \times \Gamma \times \Act \times Q \times \{\varepsilon\})$ is a finite set of \emph{internal rules}, \emph{push rules}, and \emph{pop rules},
   respectively.
The {\em size} of $\P$ is defined as $|\P|\defeq|\Gamma|+|\Act|+|\mathord{\btran{}}|$.
We write $q v \btran{a} q' w$ to mean $(q,v,a,q',w) \in \mathord{\btran{}}$.
Such a PDA $\P$ induces an LTS $\S(\P) \defeq (Q \times \Gamma^*,\Act,\{\mathord{\xrightarrow{a}} \mid a\in\Act\})$, where
$
\mathord{\xrightarrow{a}} \defeq \bigcup_{x \in \Gamma^*} \{(q v x, q' w x) \mid q v \btran{a} q' w \}
$
for each $a \in \Act$
We will abbreviate each configuration $(q,w)$ in $\S(\P)$ by $qw$; in particular  the configuration
$(q, \epsilon)$ will be denoted by just $q$.

Given a PDA $\P = (Q,\Gamma,\Act,\mathord{\btran{}})$,  $q_1, q_2 \in Q$ and $w_1,w_2\in \Gamma^\ast$
 the \emph{PDA bisimilarity problem} asks whether
$q_1 w_1 \sim q_2 w_2$ holds in $\S(\P)$.
In this paper we prove the following theorem:
\begin{theorem} \label{thm-main}
 PDA bisimilarity is nonelementary.
\end{theorem}


\section{Techniques} \label{sec-techniques}
\subsection{Large Counters} \label{sub-counters}

For each $\ell,n \ge 0$ we define $\Tow(\ell,n)$ inductively as $\Tow(0,n)\defeq n$ and
 $\Tow(\ell+1,n)\defeq 2^{\Tow(\ell,n)}$.
Let $\Omega_\ell\defeq \{0_\ell, 1_\ell\}$ be alphabets whose letters have values:
 $\val(0_\ell) = 0$ and $\val(1_\ell) = 1$.
A \emph{$(0,n)$-counter} is a word from $\Omega_0^n$.
The \emph{value} $\val(c)$ of a $(0,n)$-counter $c=\sigma_0 \cdots \sigma_{n-1}$ is defined as
 $\val(c) \defeq \sum_{i=0}^{n-1} 2^i \cdot \val(\sigma_i)$.
So the set of values $\val(c)$ of $(0,n)$-counters equals $[0,2^n-1] = [0,\Tow(1,n)-1]$.
An $(\ell,n)$-counter with $\ell \ge 1$ is a word $c = c_0 \sigma_0 c_1 \sigma_1 \cdots c_m \sigma_m$, where $m = \Tow(\ell,n)-1$,
 each $c_i$ is an $(\ell-1,n)$-counter with $\val(c_i) = i$, and $\sigma_i \in \Omega_\ell$ for $i \in [0,m]$.
We define $\val(c) \defeq \sum_{i=0}^m 2^i \cdot \val(\sigma_i)$.
Observe that $\val(c) \in [0,\Tow(\ell+1,n)-1]$ and the length of each $(\ell,n)$-counter is uniquely determined by $\ell$ and~$n$.
We call an $(\ell,n)$-counter~$c$ \emph{zero} if $\val(c) = 0$,
 and $\emph{ones}$ if $\val(c) = \Tow(\ell+1,n)-1$.
In the following we write $\Omega_{\le \ell}$ for $\bigcup_{i=0}^\ell \Omega_i$.
When $n$ is clear from the context, we may speak of an \emph{$\ell$-counter} to mean an $(\ell,n)$-counter.

\subsection{Transductions}\label{sub-transductions}
A  (real-time and non-erasing) {\em  transducer} is a tuple $T=(Q,q_0,\Sigma,\Upsilon,\delta)$,
where $Q$ is a finite set of {\em states},
$q_0\in Q$ is an {\em initial state},
$\Sigma$ and $\Upsilon$ are finite {\em alphabets},
and
$\delta:Q\times\Sigma\rightarrow Q\times\Upsilon^+$
is a transition function with output.
We say that $T$ is {\em letter-to-letter} if $\delta(q,a)\in Q\times\Upsilon$
for each $q\in Q$ and each $a\in\Sigma$.
We inductively extend $\delta$ to the function $\delta^*:Q\times\Sigma^*\rightarrow Q\times\Upsilon^*$ as follows:
 for each $w\in\Sigma^*$ and $a\in\Sigma$ we set
$\delta^*(q,\varepsilon)\defeq (q,\varepsilon)$ and
$\delta^*(q,a w)\defeq (q'', u v)$ if $\delta(q,a)=(q',u)$ and $\delta^*(q',w)=(q'',v)$.
We define the \emph{transduction} $f_T:\Sigma^*\rightarrow\Upsilon^*$ of $T$ as
$f_T(w)\defeq v$, whenever $\delta^*(q_0,w)=(q,v)$ for some $q\in Q$.
A transduction $f_T: \Sigma^*\rightarrow\Upsilon^*$ is said to be letter-to-letter if $T$ is.
We define the {\em size of $T$} as $|T|\defeq|Q|+|\Sigma|+|\Upsilon|+
\sum\{|w|: q\in Q,a\in\Sigma,\delta(q,a)=(q',w)\}$.

Given two transductions $f_1:\Sigma_1^*\rightarrow\Upsilon^*$ and
$f_2:\Sigma_2^*\rightarrow\Upsilon^*$ with $\Sigma_1\cap\Sigma_2=\emptyset$,
we define their {\em shuffle} as
$f_1 \shuf f_2:(\Sigma_1\cup\Sigma_2)^*\rightarrow \Upsilon^*$
inductively for each $w\in(\Sigma_1\cup\Sigma_2)^*$ and each $a\in(\Sigma_1\cup\Sigma_2)$
as follows:
$f_1 \shuf f_2(\varepsilon)\defeq\varepsilon$ and
$f_1 \shuf f_2(aw)\defeq f_i(a)\cdot (f_1 \shuf f_2(w))$
if $a\in\Sigma_i$, for each $i\in\{1,2\}$.

We note that from two given transducers $T_1, T_2$ with transductions
$f_{T_1}:\Sigma_1^*\rightarrow\Upsilon^*$ and
$f_{T_2}:\Sigma_2^*\rightarrow\Upsilon^*$,
one can compute in time
$O(|T_1|\cdot|T_2|)$ a transducer $T$ such that
$f_T = f_{T_1} \shuf f_{T_2}$.

We note that every non-erasing homomorphism is a transduction (witnessed by a single-state
transducer).
Having a transducer $T$,
we will often write $T$ for $f_T$ without risk of confusion.
For $w\in\Upsilon^*$, we denote by $\Sigma\mapsto w$
the homomorphism $h(a)\defeq w$ for each $a\in\Sigma$.

\subsection{Defender's forcing and Attacker's forcing} \label{sub-def-forcing}

\begin{figure}
\begin{tikzpicture}[xscale=1.6]
\tikzstyle{mystate} = [minimum height=6mm, minimum width=8mm,rounded corners,inner sep=0,draw];
\node[mystate] (s) at  (1,3)    {${}_\bullet s$};
\node[mystate] (s') at (2,3)    {$s_\bullet$};
\node[mystate] (m1) at (0.5,1.5)    {$u_1$};
\node[mystate] (m2) at (1.5,1.5)   {$u_2$};
\node[mystate] (m3) at (2.5,1.5)   {$u_3$};
\node[mystate] (t1) at (0,0)    {${}_\bullet t$};
\node[mystate] (t2) at (1,0)   {$t_\bullet$};
\node[mystate] (t3) at (2,0)    {${}_\bullet t'$};
\node[mystate] (t4) at (3,0)   {$t_\bullet'$};
\draw[->] (s) to node[left] {$a$} (m1);
\draw[->] (s) to node[left] {$a$} (m2);
\draw[->] (s) to node[right,pos=0.15] {$a$} (m3);
\draw[->] (s') to node[right,pos=0.25,xshift=-2] {$a$} (m2);
\draw[->] (s') to node[right,pos=0.25] {$a$} (m3);
\draw[->] (m1) to node[left] {$a$} (t1);
\draw[->] (m1) to node[left,yshift=-1pt] {$b$} (t3);
\draw[->] (m2) to node[left] {$a$} (t1);
\draw[->] (m2) to node[right,pos=0.05,yshift=1] {$b$} (t4);
\draw[->] (m3) to node[left,pos=0.45] {$a$} (t2);
\draw[->] (m3) to node[right,pos=0.3,xshift=-2] {$b$} (t3);
\end{tikzpicture}
\hfill
\begin{tikzpicture}[xscale=1.6]
\tikzstyle{mystate} = [minimum height=6mm, minimum width=8mm,rounded corners,inner sep=0,draw];
\node[mystate] (s) at  (1,3)    {${}_\bullet s$};
\node[mystate] (s') at (2,3)    {$s_\bullet$};
\node[mystate] (t1) at (0,0)    {${}_\bullet t$};
\node[mystate] (t2) at (1,0)   {$t_\bullet$};
\node[mystate] (t3) at (2,0)    {${}_\bullet t'$};
\node[mystate] (t4) at (3,0)   {$t_\bullet'$};
\draw[->] (s) to node[left,pos=0.3] {$a$} (t1);
\draw[->] (s) to node[left,pos=0.3] {$b$} (t3);
\draw[->] (s') to node[right,pos=0.3] {$a$} (t2);
\draw[->] (s') to node[right,pos=0.3] {$b$} (t4);
\end{tikzpicture}
\caption{(a) Or-gadget (Defender's forcing) \hspace{10mm} (b) And-gadget (Attacker's forcing)}
\label{fig-gadgets}
\end{figure}
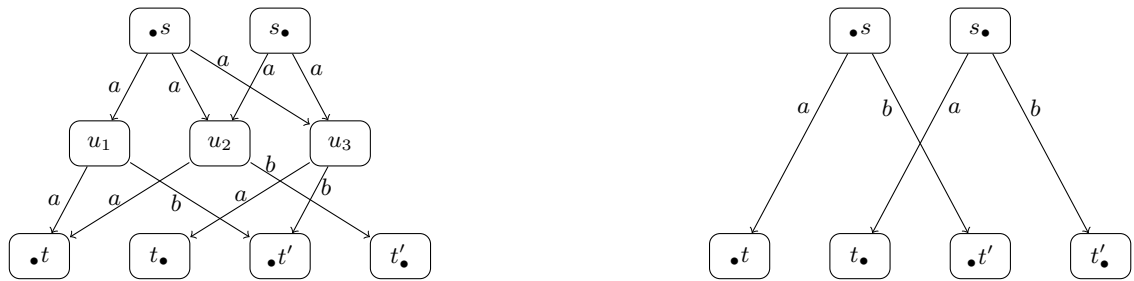
For our reduction we will use \emph{Or-gadgets} (``Defender's forcing'') and \emph{And-gadgets} (``Attacker's forcing'')
 to express logical disjunction and conjunction with bisimulation.
More precisely, we have the following lemma.
\begin{lemma}[see e.g.~\cite{JS08}] \label{lem-gadgets}
Consider the states and transitions in Figure~\ref{fig-gadgets} (a) or~(b) as part of an LTS.
The states ${}_\bullet s, s_\bullet$ may have incoming transitions, the states ${}_\bullet t,t_\bullet,{}_\bullet t',t_\bullet',$
 may have outgoing transitions (not shown).
Then we have for the gadgets in Figure~\ref{fig-gadgets}:
\begin{itemize}
 \item[(a)]
  for the Or-gadget:  ${}_\bullet s \sim s_\bullet$ if and only ${}_\bullet t \sim t_\bullet$ or ${}_\bullet t' \sim t_\bullet'$;
 \item[(b)]
  for the And-gadget:  ${}_\bullet s \sim s_\bullet$ if and only ${}_\bullet t \sim t_\bullet$ and ${}_\bullet t' \sim t_\bullet'$.
\end{itemize}
\end{lemma}
The lemma is easy to verify, see e.g.~\cite{JS08}.
In terms of a Defender-Attacker game: In the Or-gadget Defender decides if the play continues in $({}_\bullet t, t_\bullet)$ or in $({}_\bullet t', t_\bullet')$,
 whereas in the And-gadget it is Attacker who decides this.

\subsection{Macro Rules}

We will construct a PDA with many control states and rules.
In the interest of succinctness and readability we use \emph{macro rules}
 that compactly represent a set of PDA transition rules with a certain role.
For the rest of the section, fix a PDA $\P=(Q,\Gamma,\Act,\mathord{\btran{}})$ with $a,b \in \Act$.

\subsubsection{Macro rules with one state on the left-hand side.}

For $p,q \in Q$, $\sigma\in\Gamma$ and $a_1\cdots a_\ell\in\Act^\ell$ with $\ell \ge 1$ we write
\[
 p\sigma \btran{a_1 \cdots a_\ell} q
\]
to denote that there are $p_1, \ldots, p_\ell \in Q$ with
$
 p\sigma \btran{a_1} p_1 \btran{a_2} p_2 \cdots \btran{a_\ell} p_\ell = q
$,
and there are no other rules with $p \sigma$ on the left-hand side and no other rules involving $p_1, \ldots, p_{\ell-1}$.

For $p,q \in Q$ and a regular language $L\subseteq\Gamma^*$ and a
transduction $T:\Gamma^*\rightarrow\Act^*$ we write
\[
 p L\btran{T} q
\]
to denote that $\P$ contains control states and rules described below.
These rules make sure that when $\P$ is in a configuration $p y$ for $y \in \Gamma^*$,
 then the shortest prefix $w$ of~$y$ with $w \in L$ will be popped, and $\# T(w) \#$ will be read (where $\# \in \Act$ is a special action symbol),
 and the control state will be changed to~$q$;
if $y$ does not have a prefix~$w$ with $w \in L$, then $y$ will be popped, and $\# T(y)$ will be output.
This behaviour is the result of a product construction between the minimal
 deterministic finite automaton (DFA) accepting~$L$ and the transducer~$T$.
More precisely, let $A=(Q_A,q_{0}^A,\Gamma,F_A,\delta_A)$ be the minimal DFA
that accepts $L$, where
$Q_A$ is the finite set of states,
$q_0^A \in Q_A$ is the initial state,
$F_A \subseteq Q_A$ is the set of final states,
$\delta_A:Q_A\times\Gamma\rightarrow Q_A$
is the transition function. 
Assume $T=(Q_T,q_{0}^T,\Gamma,\Act,\delta_T)$.
Then $\P$ contains the control states $Q_A\times Q_T$ and the following rules:
\begin{itemize}
\item $p \btran{\#}(q_{0}^A,q_{0}^T)$;
\item $(q^A,q^{T})\btran{\#} q$
for each $q^A\in F_A$ and each $q^T\in Q_T$;
\item for each $\sigma\in\Gamma$, each $q^A\in Q_A\setminus F_A$ and each $q^T\in Q_T$, where $\delta_T(q^T,\sigma)=(r^T,w)$,
 we have the (macro) rule $(q^A,q^T) \sigma \btran{w} (\delta_A(q^A,\sigma),r^T)$.
\end{itemize}
There are no other rules with $p$ on the left-hand side, and no other rules involving $Q_A\times Q_T$.

If $p L \btran{T} q$ and $w \in L$ but no proper prefix of~$w$ is in~$L$,
 then we have in the LTS $\S(\P)$
\begin{equation}
 \text{for all } x \in \Gamma^* : \quad
 p w x \xrightarrow{\#} s_0 \xrightarrow{a_1} s_1 \xrightarrow{a_2} \ldots \xrightarrow{a_\ell} s_\ell \xrightarrow{\#} q x\,,
 \label{eq-trans}
\end{equation}
where the path is deterministic, $T(w) = a_1 \cdots a_\ell$, and $s_0, \ldots, s_\ell$ are configurations of~$\P$, i.e., states in $\S(\P)$.


We will need the following lemma, which shows how to compare two counters in terms of their images of two given transducers.
For the statement of the lemma, recall the concept of counters and the alphabets $\Omega_\ell$ from Section~\ref{sub-counters} and recall for each alphabet $\Omega$ and each word $w$ we denote by $\Omega\mapsto w$
the homomorphism
that maps every element from $\Omega$ to $w$.
\begin{lemma} \label{lem-transducers}
Let $T_1, T_2 : \Omega_{\le \ell+1}^* \to \Act^*$ be letter-to-letter transducers for some $\ell \ge 0$.
Let $\P=(Q,\Gamma,\Act,\mathord{\btran{}})$ be a PDA with
$\{
{}_\bullet p,p_\bullet,
{}_\bullet q,q_\bullet,
{}_\bullet r,r_\bullet\}\subseteq Q$,
 $\Omega_{\le \ell+1}\subseteq\Gamma$ and the following macro rules:
\begin{align*}
  {}_\bullet q \left( \Omega_{\le \ell-1}^*\cdot \Omega_\ell \right)^*\cdot \Omega_{\ell+1}
          & \btran{
              T_1} {}_\bullet r
             && \text{apply $T_1$} \\
  q_\bullet \left( \Omega_{\le \ell-1}^*\cdot \Omega_\ell \right)^*\cdot \Omega_{\ell+1}
          & \btran{
              T_2} r_\bullet
             && \text{apply $T_2$} \\
  {}_\bullet p\ \Omega_{\le \ell}^*\cdot \Omega_{\ell+1}\cdot
                        \Omega_{\le \ell}^*\cdot \Omega_{\ell+1}
                      & \btran{\Gamma\mapsto a} {}_\bullet q  && \begin{aligned} & \text{pop two $\ell$-counters} \\
                                                                          & \text{and two $\Omega_{\ell+1}$-symbols} \end{aligned} \\
  p_\bullet\  \Omega_{\le \ell}^*\cdot \Omega_{\ell+1}
                      & \btran{\Gamma \mapsto aa} q_\bullet && \begin{aligned} & \text{pop one $\ell$-counter} \\
                                                                          & \text{and one $\Omega_{\ell+1}$-symbol} \end{aligned} \\
\end{align*}
Let $\sigma_1, \sigma_2, \sigma_3 \in \Omega_{\ell+1}$,
and let $w_1, w_2, w_3$ be $\ell$-counters.
\begin{itemize}
 \item[(a)]
  Assume $x_1, x_2 \in \Gamma^*$ such that ${}_\bullet r x_1 \sim r_\bullet x_2$.
  Then
  \[
   {}_\bullet q w_1 \sigma_1 x_1 \sim q_\bullet w_2 \sigma_2 x_2 \ \Longleftrightarrow \
   T_1(w_1 \sigma_1) = T_2(w_2 \sigma_2)\,.
  \]
 \item[(b)]
  Assume $x \in \Gamma^*$ such that ${}_\bullet r x \sim r_\bullet w_1 \sigma_1 x$.
  Then
  \[
   {}_\bullet p w_3 \sigma_3 w_2 \sigma_2 w_1 \sigma_1 x \sim p_\bullet w_3 \sigma_3 w_2 \sigma_2 w_1 \sigma_1 x \ \Longleftrightarrow \
    T_1(w_1 \sigma_1) = T_2(w_2 \sigma_2)\,.
  \]
\end{itemize}
\end{lemma}
\begin{proof} 
Part~(a) is immediate from the definitions.
For part~(b) we have:
\begin{align*}
                       & {}_\bullet p w_3 \sigma_3 w_2 \sigma_2 w_1 \sigma_1 x \sim p_\bullet w_3 \sigma_3 w_2 \sigma_2 w_1 \sigma_1 x \\
  \quad\Longleftrightarrow\quad & {}_\bullet q w_1 \sigma_1 x \sim q_\bullet w_2 \sigma_2 w_1 \sigma_1 x && \text{by the first two rules} \\
  \quad\Longleftrightarrow\quad & T_1(w_1 \sigma_1) = T_2(w_2 \sigma_2) && \text{by part~(a)}
\end{align*}
\qed
\end{proof}

\subsubsection{Macro rules with a state pair on the left-hand side.}

\renewcommand{\vec}[1]{\overrightarrow{#1}}

In the following we assume that control states, i.e., the elements of~$Q$, are of the form
${}_\bullet q$ and $q_\bullet$.
By~$\vec{q}$ we refer to the \emph{state pair} $({}_\bullet q,q_\bullet) \in Q^2$.
Given $w \in \Gamma^*$, we write $\Bis{\vec{q} w}$ to denote that ${}_\bullet q w \sim q_\bullet w$.

For $\sigma_1 \cdots \sigma_\ell \in \Gamma^\ell$ with $\ell \ge 0$ we write
\[
 \vec{q} \ctran{} \vec{r} \sigma_1 \sigma_2 \cdots \sigma_\ell
\]
to denote that there are state pairs ${\vec{q_0}}, \ldots, {\vec{q_\ell}}$ with
${\vec{q_\ell}} = \vec{q}$ and
\begin{align*}
 & {}_\bullet q_\ell \btran{a} {}_\bullet q_{\ell-1} \sigma_\ell, \ \ldots, \
  {}_\bullet q_1 \btran{a} {}_\bullet q_0 \sigma_1, \
{}_\bullet q_0  \btran{a} {}_\bullet r \quad \text{and} \\
 &  q_{\ell\bullet} \btran{a} q_{\ell-1\bullet} \sigma_\ell, \ \ldots, \
  q_{1\bullet} \btran{a} q_{0\bullet} \sigma_1, \
q_{0\bullet}  \btran{a} r_\bullet,
\end{align*}
and there are no other rules with ${}_\bullet q$ or $q_\bullet$ on the left-hand side
 and no other rules involving ${\vec{q_0}}, \ldots, {\vec{q_{\ell-1}}}$.
With this macro rule we have
\begin{equation}
 \text{for all } x \in \Gamma^* : \quad \Bis{\vec{q}x} \ \Longleftrightarrow \ \Bis{\vec{r} \sigma_1 \sigma_2 \cdots \sigma_\ell x}\,. \label{eq-circ}
\end{equation}
For $\sigma_1, \sigma_2 \in \Gamma \cup \{\varepsilon\}$ we write
\[
  \vec{q} \ctran{\Def} \left\{ \vec{r_1} \sigma_1, \ \vec{r_2} \sigma_2 \right\}
\]
to denote that Defender's forcing is implemented as described in Lemma~\ref{lem-gadgets}~(a);
i.e., in terms of Figure~\ref{fig-gadgets}~(a) we have the state correspondences
 ${}_\bullet s = {}_\bullet q$ and $s_\bullet = q_\bullet$
 and ${}_\bullet t = {}_\bullet r_1\sigma_1$ and $t_\bullet = r_{1\bullet}\sigma_1$
 and ${}_\bullet t' = {}_\bullet r_2\sigma_2$ and $t_\bullet' = r_{2\bullet}\sigma_2$,
the internal rules ${}_\bullet q \btran{a} u_1, \ \ldots, \ q_\bullet \btran{a} u_3$
and finally the push rules $u_1 \btran{a} {}_\bullet r_1\sigma_1, \ \ldots, \ u_3 \btran{b}
{}_\bullet r_2\sigma_2$,
as prescribed by Figure~\ref{fig-gadgets}~(a).
Intuitively, in a Defender-Attacker game, when the play is in a configuration $({}_\bullet qx,
q_\bullet x)$ for $x \in \Gamma^*$,
 then Defender chooses whether the game will be in
 $({}_\bullet r_1\sigma_1 x, r_{1\bullet}\sigma_1 x)$ or in
 $({}_\bullet r_2\sigma_2 x, r_{2,\bullet}\sigma_2 x)$.
In other words, we have $\Bis{\vec{q}}x$ iff $\Bis{\vec{r_1}\sigma_1 x}$ or $\Bis{\vec{r_2}\sigma_2 x}$.
We generalise this notation to sets:
for $\{w_1, \ldots, w_\ell\} \subseteq \Gamma^*$ we also write
\[
 \vec{q} \ctran{\Def} \left\{ \vec{r_1} w_1, \ \ldots, \ \vec{r_\ell} w_\ell \right\}
\]
to denote that a sequence of Or-gadgets (Figure~\ref{fig-gadgets}~(a)) is used to achieve
\begin{equation}
 \text{for all } x \in \Gamma^* : \quad \Bis{\vec{q}x}\ \quad\Longleftrightarrow\quad \bigvee_{i=1}^\ell
\Bis{\vec{r_i}w_i} x\,. \label{eq-def}
\end{equation}

\noindent
Similarly we write
\[
 \vec{q} \ctran{\Att} \left\{ \vec{r_1} w_1, \ \ldots, \ \vec{r_\ell} w_\ell \right\}
\]
to denote that And-gadgets (Attacker's forcing, Figure~\ref{fig-gadgets}~(b)) are used to achieve
\begin{equation}
 \text{for all } x \in \Gamma^* : \quad \Bis{\vec{q}x} \ \Longleftrightarrow \ \bigwedge_{i=1}^\ell
\Bis{\vec{r_i}w_i x}\,. \label{eq-att}
\end{equation}


\section{The Construction} \label{sec-construction}
\newcommand{\Last}{\mathsf{last}}
We prove Theorem~\ref{thm-main} by showing that PDA bisimilarity is $k$-EXPSPACE-hard for all $k \ge 1$.

As the first step of our reduction we consider a problem on letter-to-letter transducers.
A \emph{transducer machine} is a triple $\mathcal{T} = (\ell,T_1,T_2)$,
 where $\ell \ge 1$, and $T_1,T_2:\{0,1\}^*\to\Upsilon^*$ are letter-to-letter transducers.
Given a transducer machine~$\T$ we call $z\in\{0,1\}^\ell$ a {\em dead end} if
there is no $z'\in\{0,1\}^\ell$ with $T_1(z)=T_2(z')$.
We say that $\T$ is \emph{deterministically terminating}
 if there are $t \in \N$ and words $z_0, \ldots, z_t \in \{0,1\}^\ell$ such that
\begin{itemize}
\item $z_0 = 1^\ell$,
\item for each $z_i$ there is at most one $z' \in \{0,1\}^\ell$ with $T_1(z_i) = T_2(z')$,
\item $T_1(z_i) = T_2(z_{i+1})$ holds for all $i \in [0,t-1]$, and
\item $z_t$ is a dead end.
\end{itemize}
If $\T$ is deterministically terminating we define $\Last(\T)\defeq z_t$.
The first step of our reduction is applying the algorithm of the following proposition.
\begin{proposition} \label{prop-trans-problem}
For each $k\geq 1$ there exists a $k$-EXPSPACE-complete language $L\subseteq\Sigma^*$ such that
the following is computable in polynomial time:

INPUT: $w\in\Sigma^n$.

OUTPUT:
Transducers $T_1, T_2 : \{0,1\}^*\to\Upsilon^*$, where
 $\mathcal{T} = (\Tow(k,n),T_1,T_2)$ is a
deterministically terminating transducer machine and
$0^{\Tow(k,n)}$ is a dead end
such that moreover
$$
w\in L \qquad\text{if and only if}\qquad \Last(\T)=0^{\Tow(k,n)}\,.
$$
\end{proposition}
\begin{proof}[sketch]
\newcommand{\enc}{\mathit{enc}}%
Let us fix some $k\geq 1$. Let us first mention how the language $L$ can be chosen.
The following claim is a simple adaption of the linear speedup theorem, we refer the reader to \cite{Pap94} for details.
\medskip

\noindent
{\bf Claim: }
For each $k\geq 1$ there exists a deterministic Turing machine (DTM) $\mathcal{M}$
and some $m\in\N$ such that the following holds:
\begin{itemize}
\item[(1)] $m$ is the sum of the number of states of $\mathcal{M}$ plus the number of tape symbols of $\mathcal{M}$.
\item[(2)] the DTM $\mathcal{M}$ is $\left\lfloor\frac{\Tow(k,n)}{n+m}\right\rfloor-n$ space bounded
for all but finitely many $n\in\N$.
\item[(3)] For each $n\geq 0$ we have
\begin{itemize}
\item[(i)] the DTM $\mathcal{M}$ has a unique accepting configuration and a unique rejecting configuration that both do not have any successor configuration
on each input of length $n$ and
\item[(ii)]
  the unique computation of~$\mathcal{M}$ on $w$ either reaches the accepting or the
rejecting configuration for each word $w$ of length $n$.
\end{itemize}
\item[(4)] The acceptance problem of $\mathcal{M}$ is complete
for $k$-EXPSPACE under polynomial time many-one reductions.
\end{itemize}

Let us fix some $k\ge 1$ and some Turing machine $\mathcal{M}$ that satisfies points
(1) to (4) of the above Claim for the rest of this proof.
We define $L\defeq L(\mathcal{M})$.
Let us assume that the input alphabet of $\mathcal{M}$ is $\Sigma$.
Let us fix an input word $w \in \Sigma^n$.
We can assume without loss of generality that
$\mathcal{M}$ is $\left\lfloor\frac{\Tow(k,n)}{n+m}\right\rfloor-n$ space bounded
on input $w$ (the other finitely many cases can be dealt with explicitly in our reduction).
Recall that by Point (3) of the above Claim we have that
$\mathcal{M}$ has a unique accepting configuration and a unique rejecting configuration,
 and that all computations of~$\mathcal{M}$ reach either the accepting or the rejecting configuration,
 and that the accepting and rejecting configurations have no successor configurations.
In a first step, we modify $\mathcal{M}$ to a DTM $\mathcal{M}_w$ so that $\mathcal{M}_w$ started
on the empty tape first writes~$w$ on the tape
 and then simulates~$\mathcal{M}$ on~$w$.
We note that we can construct $\mathcal{M}_w$ in such a way that
\begin{itemize}
\item $\mathcal{M}_w$ can be obtained
from $\mathcal{M}$ by adding at most $n$ additional states and corresponding transitions
that allow us to initially copy $w$ onto the working tape,
\item the sum of the number of states of $\mathcal{M}_w$ plus the number of tape symbols
of $\mathcal{M}_w$ is $n+m$,
\item $\mathcal{M}_w$ is $\left\lfloor\frac{\Tow(k,n)}{n+m}\right\rfloor-n+n=
\left\lfloor\frac{\Tow(k,n)}{n+m}\right\rfloor$ space bounded and
\item $w \in L(\mathcal{M})$ if and only if $\varepsilon \in L(\mathcal{M}_w)$.
\end{itemize}

Fix a binary encoding~$\enc$ of configurations so that each
tape symbol and each pair consisting of a state and a tape symbol of $\mathcal{M}_w$ can
 be (injectively) encoded by
a binary string of length $n+m$.
We extend this encoding to configurations of $\mathcal{M}_w$ by mapping each configuration $c$
(injectively) to a string $\enc(c) \in \{0,1\}^\ell$
 where $\ell = \Tow(k,n)$.
Moreover we assume that the initial configuration of~$\mathcal{M}_w$ (with empty tape)
is encoded by~$1^\ell$
 and that the (unique) accepting configuration of~$\mathcal{M}_w$ is encoded by~$0^\ell$.
It remains to argue that one can construct transducers $T_1, T_2 : \{0,1\}^*\to\Upsilon^*$ so that
\begin{itemize}
\item[(*)]
 for all configurations $c, c'$ of~$\mathcal{M}_w$ we have $T_1(\enc(c)) = T_2(\enc(c'))$ if and only if
  $c'$ is a successor configuration (i.e., the unique one) of~$c$.
\end{itemize}
For establishing (*), the idea is to construct~$T_1, T_2$ so that
 if $c$ is a configuration of~$\mathcal{M}_w$
 then $T_1(\enc(c))$ is an encoding of~$c'$, where $c'$ is the successor configuration of~$c$,
 and $T_2(\enc(c))$ is an encoding of~$c$.
The most straightforward implementation of this idea would be to let $T_1(\enc(c)) = \enc(c')$ and $T_2(\enc(c)) = \enc(c)$.
However, this cannot be easily done, if at all, loosely speaking because
 the read-write head of~$\mathcal{M}_w$ may move in the direction ``opposite'' to the transducers so that the transducer~$T_1$ would have to ``guess''
 where the read-write head is before it actually sees it.
Therefore, we construct $T_1, T_2$ so that their output is ``delayed'' by a few steps:
Transducer $T_1$ remembers in its finite control the last few bits of the encoded tape and outputs the bits of the encoding of the successor configuration only
 after $T_1$ can be sure about them.
Transducer $T_2$ does not compute the successor configuration, but only re-encodes the encoded configuration,
 and outputs the bits of the new encoding in a similarly delayed way as~$T_1$.
Since transducers need to output a single symbol per step, the transducers $T_1, T_2$ output a dummy symbol in the first few steps.
At the end they need to output a single symbol containing the last few bits of the new encoding.
As a consequence the alphabet $\Upsilon$ cannot be (easily) taken to be binary;
thus we simply choose $\Upsilon$ sufficiently large for this to work.
\qed
\end{proof}

Let us fix some $k\geq 1 $ for the rest of this section and let us fix the
$k$-EXPSPACE complete language $L\subseteq\Sigma^*$ that satisfies Proposition \ref{prop-trans-problem}.
Moreover let $w\in\Sigma^n$ be a word.
Our overall goal is to compute from $w$ in polynomial time a PDA and two of its configurations that are
bisimilar if and only if $w\in L$.
As an intermediate step, let us fix for the rest of this section the output $(T_1,T_2)$ of the algorithm of Proposition~\ref{prop-trans-problem} on input $w$,
 and let $\T = (\Tow(k,n),T_1,T_2)$.
In the rest of the section we will show how to compute from~$T_1, T_2$ and~$n$ in time
 polynomial in $|T_1|+|T_2|+n$
a PDA $\P=(Q,\Omega,\Act,\mathord{\btran{}})$ so that we have
 $\Bis{\q{start}}$ in~$\S(\P)$ if and only if $\Last(\T)=0^{\Tow(k,n)}$ holds in~$\T$,
where ${}_\bullet\langle\mathsf{start}\rangle$ and
$\langle\mathsf{start}\rangle_\bullet$ will be
control states of $\P$.
We recall that $k$ is a fixed constant.

\noindent
Let
\begin{align*}
B \defeq
& \{\textsf{start}, \textsf{stop$_\ell$},
\textsf{testDec$_\ell$}, \textsf{testDec$^1_\ell$},
\textsf{ones$_\ell$}, \textsf{ones$^1_\ell$},
\textsf{decOk$_\ell$},
\textsf{zero$_\ell$}, \textsf{zero$^1_\ell$},
\textsf{dec$_\ell$}, \textsf{dec$^1_\ell$},
\textsf{dec$^{(i)}_0$}, \\
& \ \textsf{fin},
\textsf{testFin},
\textsf{popAll},
\textsf{next},
\textsf{next$^1$},
\textsf{tran},
\textsf{testTran},
\textsf{testTran$^1$}
\mid 0 \le \ell \le k+2, \ 1 \le i \le n\} \cup B_{\mathit{impl}}
\end{align*}
 be a set of ``basic symbols'' that we use to construct the control states~$Q$.
The set $B_{\mathit{impl}}$ contains further (implicit) symbols that are needed to implement macro rules.
In the following we regard each element of~$B$ as a single symbol (of length~$1$).
Define
$\Omega\defeq\Omega_{\le k+1}$,
$\Act\defeq\{0,1,\#,a,b\}\uplus\Upsilon$,
 $Q = {}_\bullet B \cup B_\bullet$, where
 ${}_\bullet B \defeq \{ \ql{$\alpha$} \mid \alpha \in B^*, \ 1 \le |\alpha| \le k+2\}$ and similarly
 $B_\bullet \defeq \{ \qr{$\alpha$} \mid \alpha \in B^*, \ 1 \le |\alpha| \le k+2\}$.
For instance, we have $\ql{dec$_{k-1}$ ones$_k$} \in Q$.
We will use $\textsf{\sta}$ to indicate an arbitrary word $\alpha \in B^*$ with $1 \le |\alpha| \le k+1$.

This section is organised as follows.
In Section \ref{sec-consecutive} we show how we can implement a bisimulation (sub-)game in $\S(\P)$
that allows us to test whether two $\ell$-counters 
have consecutive values for each $\ell\geq 0$.
In Section \ref{sec-building} we show how we can implement a bisimulation (sub-)game in $\S(\P)$
that allows Defender to push an $\ell$-counter onto the stack for each $\ell\geq 0$.
We conclude our reduction in Section \ref{sec-simulating}.

\subsection{Checking Counters for Consecutive Values}{\label{sec-consecutive}}

For each $\ell \in [0, k]$
 we include control states $\ql{stop$_\ell$}, \qr{stop$_\ell$}$
 such that for all $x \in \Omega^*$ and all $\sigma \in \Omega_{\ell+1}$ and all $\ell$-counters~$w$ we have
\begin{equation}
 \ql{stop$_\ell$} x \sim \qr{stop$_\ell$} w \sigma x \label{eq-stop}
\end{equation}
This is easily achieved, for instance by including no rules with $\ql{stop$_\ell$}$ or $\qr{stop$_\ell$}$ on the left-hand side.

We need to be able to verify whether two counters (at convenient positions) on the stack
have consecutive values.
To this end we include rules such that the following statement holds:

\begin{lemma} \label{lem-check-counters}
Let $x \in \Omega^*$, and $\sigma_1, \sigma_2, \sigma_3 \in \Omega_{\ell+1}$,
and let $w_1, w_2, w_3$ be $\ell$-counters.
Then $\Bis{\q{testDec$_\ell$} w_3 \sigma_3 w_2 \sigma_2 w_1 \sigma_1 x}$
 iff $\val(w_1) = \val(w_2) + 1$.
\end{lemma}

Let $T^{+0}_\ell, T^{+1}_\ell : (\Omega_\ell \cup \Omega_{\ell+1})^* \to \{0,1,a,b\}^*$ be the transducers depicted in Figure~\ref{fig-transducers}.
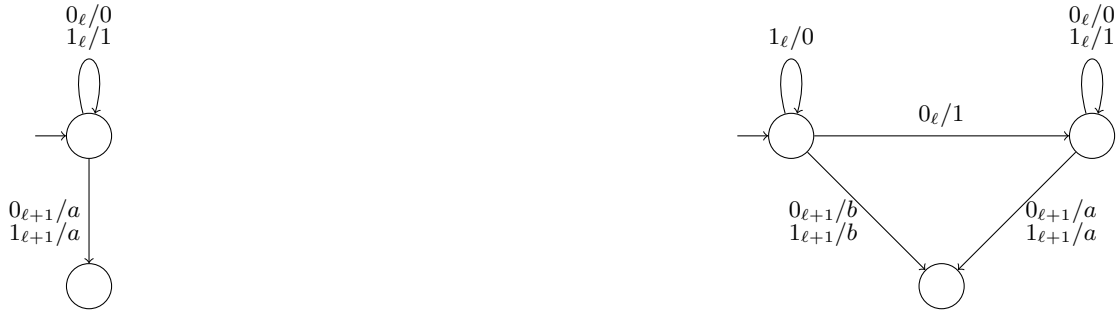
\begin{figure}
\begin{tikzpicture}[every state/.style={minimum size=6mm}]
 \node[initial, initial text=,state] (0) at (0,0) {$$};
 \node[state] (2) at (0,-2) {$$};
 \path[->] (0) edge [loop above,min distance=10mm] node {$1_\ell/1$} node[yshift=3mm] {$0_\ell/0$} (0)
               edge [left] node {$0_{\ell+1}/a$} node[yshift=-3mm] {$1_{\ell+1}/a$} (2);
\end{tikzpicture}
\hfill
\begin{tikzpicture}[every state/.style={minimum size=6mm}]
 \node[initial, initial text=,state] (0) at (0,0) {$$};
 \node[state] (1) at (4,0) {$$};
 \node[state] (2) at (2,-2) {$$};
 \path[->] (0) edge [loop above,min distance=10mm] node {$1_\ell/0$} (0)
               edge [above] node {$0_\ell/1$} (1)
               edge [left] node {$0_{\ell+1}/b$} node[yshift=-3mm] {$1_{\ell+1}/b$} (2)
           (1) edge [loop above,min distance=10mm] node {$1_\ell/1$} node[yshift=3mm] {$0_\ell/0$} (1)
               edge [right] node {$0_{\ell+1}/a$} node[yshift=-3mm] {$1_{\ell+1}/a$} (2);
\end{tikzpicture}
\caption{Transducers $T^{+0}_\ell$ and $T^{+1}_\ell$}
\label{fig-transducers}
\end{figure}
Transducers~$T^{+0}_\ell, T^{+1}_\ell$ interpret the input word over $\Omega_{\ell}$ as a number in binary,
 with the least significant bit read first.
Transducer~$T^{+0}_\ell$ copies the number and outputs $a$ upon reading an $\Omega_{\ell+1}$-symbol.
Transducer~$T^{+1}_\ell$ attempts to increase the number by~$1$ and output~$a$ upon reading an $\Omega_{\ell+1}$-symbol,
 but it outputs~$b$ if the input number consisted only of~$1$s.
If $w_1, w_2$ are $\ell$-counters and $\sigma_1, \sigma_2 \in \Omega_{\ell+1}$, then we have:
\begin{equation}
\begin{aligned}
                        & (T^{+0}_\ell \shuf \Omega_{\le \ell-1} \mapsto a)(w_1 \sigma_1) = (T^{+1}_\ell \shuf \Omega_{\le \ell-1} \mapsto a)(w_2 \sigma_2) \\
 \quad \text{iff} \quad & \val(w_1) = \val(w_2)+1  \label{eq:plus-one-trans}
\end{aligned}
\end{equation}
Transducers $T^{+0}_\ell, T^{+1}_\ell$ are used in the following rules.

\begin{align*}
  \ql{testDec$^1_\ell$} \left( \Omega_{\le \ell-1}^*\cdot \Omega_\ell \right)^*\cdot \Omega_{\ell+1}
         & \btran{
            T^{+0}_\ell \shuf \Omega_{\le \ell-1} \mapsto a} \ql{stop$_\ell$} \\
  \qr{testDec$^1_\ell$} \left( \Omega_{\le \ell-1}^*\cdot \Omega_\ell \right)^*\cdot \Omega_{\ell+1}
         & \btran{
            T^{+1}_\ell \shuf \Omega_{\le \ell-1} \mapsto a} \qr{stop$_\ell$}\\
  \ql{testDec$_\ell$} \Omega_{\le \ell}^*\cdot \Omega_{\ell+1}\cdot
                        \Omega_{\le \ell}^*\cdot \Omega_{\ell+1}
                      & \btran{\Omega \mapsto a} \ql{testDec$^1_\ell$} \\
  \qr{testDec$_\ell$} \Omega_{\le \ell}^*\cdot \Omega_{\ell+1}
                      & \btran{\Omega \mapsto aa} \qr{testDec$^1_\ell$}
\end{align*}

\begin{proof}[of Lemma~\ref{lem-check-counters}]
By~\eqref{eq-stop} we can apply Lemma~\ref{lem-transducers}~(b).
Hence we have:
\begin{align*}
                       & \Bis{\q{testDec$_\ell$} w_3 \sigma_3 w_2 \sigma_2 w_1 \sigma_1 x} \\
  \Longleftrightarrow\ & (T^{+0}_\ell \shuf \Omega_{\le \ell-1} \mapsto a)(w_1 \sigma_1) = (T^{+1}_\ell \shuf \Omega_{\le \ell-1} \mapsto a)(w_2 \sigma_2)
                           && \text{by Lemma~\ref{lem-transducers}~(b)} \\
  \Longleftrightarrow\ & \val(w_1) = \val(w_2)+1 && \text{by~\eqref{eq:plus-one-trans}}
\end{align*}
\qed
\end{proof}

\subsection{Building Counters}{\label{sec-building}}

In the lemmas below we will make statements about properties of $\S(\P)$
if we include certain rules to~$\P$.
For better readability, we will state the properties before we list the rules.

Lemma~\ref{lem-build-counters} below demonstrates how to construct large counters.
Recall that we use $\textsf{\sta}$ to indicate an arbitrary word
from $B^*$ with length between $1$ and $k+1$.
We include rules such that the following holds:

\begin{lemma} \label{lem-build-counters}
Let $\ell \in [0,k]$, and $x \in \Omega^*$, and $\sigma, \tau \in \Omega_{\ell+1}$, and $v, w$ be $\ell$-counters.
\begin{itemize}
\item[(a)]
 Then $\Bis{\q{ones$_\ell$ \sta} x}$ iff $\Bis{\q{\sta} w x}$, where $w$ is the ones $\ell$-counter.
\item[(b)]
 Then $\Bis{\q{decOk$_\ell$ \sta} v \sigma w \tau x}$ iff
      $\val(v) + 1 = \val(w)$ and $\Bis{\q{\sta} v \sigma w \tau x}$.
\item[(c)]
 Then $\Bis{\q{zero$_\ell$ \sta} \sigma w \tau x}$ iff
      $\val(w) = 1$ and $\Bis{\q{\sta} v \sigma w \tau x}$, where $v$ is the $\ell$-counter with $\val(v) = 0$.
\item[(d)]
 Then $\Bis{\q{dec$_\ell$ \sta} \sigma w \tau x}$ iff
      $\val(w) \ne 0$ and $\Bis{\q{\sta} v \sigma w \tau x}$, where $v$ is the $\ell$-counter with $\val(v) + 1 = \val(w)$.
\end{itemize}
\end{lemma}

The following rules for the special case $\ell = 0$ are included.

\begin{align*}
  \q{ones$_0$ \sta} & \ctran{}     \q{\sta} 1_0^n            && \text{push a ones $0$-counter} \\
  \q{decOk$_0$ \sta} & \ctran{\Att} \left\{ \q{\sta} , \right.               && \begin{aligned} &\text{believe that the values of the} \\
                                                                                             &\text{top two $0$-counters differ by $1$} \end{aligned} \\
                     & \hspace{10mm}\left.  \q{testDec$_0$} 0_0^n0_{1} \right\}   && \begin{aligned} &\text{OR challenge that claim by} \\
                                                                                             &\text{invoking \textsf{testDec$_0$}} \end{aligned}\\
  \q{zero$_0$ \sta} & \ctran{}     \q{decOk$_0$ \sta} 0_0^n  && \begin{aligned} &\text{push a zero $0$-counter and check} \\
                                                                          &\text{if it is over a $0$-counter with value~$1$} \end{aligned}\\
  \q{dec$_0$ \sta}  & \ctran{\Def}
\left\{ \q{dec$_0^{(1)}$\sta}0_0,
\q{dec$_0^{(1)}$\sta}1_0\right\}&&
\text{push the first bit of the}\\[-0.3cm]
&&&\text{decremented $0$-counter} \\
\forall 1\leq i<n:\q{dec$_0^{(i)}$\sta}  & \ctran{\Def}
\left\{ \q{dec$_0^{(i+1)}$\sta}0_0,
\q{dec$_0^{(i+1)}$\sta}1_0\right\}&&
\text{push the $(i+1)^{\text{st}}$ bit of the}\\[-0.3cm]
&&&\text{decremented $0$-counter} \\
\q{dec$_0^{(n)}$\sta}  & \ctran{}
\q{decOk$_0$ \sta} &&
\text{verify if the $0$-counter has been}\\[-0.1cm]
&&&\text{correctly decremented}\\
\end{align*}

The following rules are included for $1\leq \ell \le k$.

\begin{align*}
\q{ones$_\ell$ \sta}   & \ctran{}     \q{ones$_{\ell-1}$  ones$_\ell^1$ \sta} 1_\ell && \text{push $1_\ell$ and a ones $(\ell-1)$-counter}\\
\q{ones$_\ell^1$ \sta} & \ctran{\Def} \left\{ \q{dec$_{\ell-1}$   ones$_\ell^1$ \sta} 1_\ell , \right.  && \text{push $1_\ell$ and a decremented $(\ell-1)$-counter} \\
                      & \hspace{10mm}\left.  \q{zero$_{\ell-1}$ \sta} 1_\ell         \right\} && \text{OR push $1_\ell$ and a zero $(\ell-1)$-counter} \\
  \q{decOk$_\ell$ \sta} & \ctran{\Att} \left\{ \q{\sta} , \right.               && \begin{aligned} &\text{believe that the values of the} \\
                                                                                             &\text{top two $\ell$-counters differ by $1$} \end{aligned} \\
                     & \hspace{10mm}\left.  \q{ones$_\ell$ testDec$_\ell$} 0_{\ell+1} \right\}   && \begin{aligned} &\text{OR challenge that claim by} \\
                                                                                             &\text{invoking \textsf{testDec$_\ell$}} \end{aligned} \\
  \q{zero$_\ell$ \sta}   & \ctran{}     \q{ones$_{\ell-1}$  zero$_\ell^1$ \sta} 0_\ell && \text{push $0_\ell$ and a ones $(\ell-1)$-counter}\\
  \q{zero$_\ell^1$ \sta} & \ctran{\Def} \left\{ \q{dec$_{\ell-1}$   zero$_\ell^1$ \sta} 0_\ell , \right.  && \text{push $0_\ell$ and a decr.\ $(\ell-1)$-counter} \\
                      & \hspace{10mm}\left.  \q{zero$_{\ell-1}$ decOk$_\ell$ \sta} 0_\ell         \right\} && \text{OR push $0_\ell$ and a zero $(\ell-1)$-counter} \\
\q{dec$_\ell$ \sta}   & \ctran{\Def} \left\{ \q{ones$_{\ell-1}$  dec$_\ell^1$ \sta} \sigma \mid \sigma \in \Omega_\ell \right\}  && \text{push from $\Omega_\ell$ and a ones $(\ell-1)$-counter} \\
\q{dec$_\ell^1$ \sta} & \ctran{\Def} \left\{ \q{dec$_{\ell-1}$   dec$_\ell^1$ \sta} \sigma , \right.  && \text{push from $\Omega_\ell$ and a decr.\ $(\ell-1)$-counter} \\
                     & \hspace{10mm}\left.  \q{zero$_{\ell-1}$  decOk$_\ell$ \sta} \sigma \mid \sigma \in \Omega_\ell \right\} && \text{OR push from $\Omega_\ell$ and a zero $(\ell-1)$-counter} \\
\end{align*}


\begin{proof}[of Lemma~\ref{lem-build-counters}]
The proof is by induction on~$\ell$.

\noindent
{\em Induction base.}

Case (a) immediately follows from the rule for $\q{ones$_0$ \sta}$.

For part (b) we have:

\begin{align*}
\Bis{\q{decOk$_{0}$ \sta} v\sigma w\tau x}
& \quad\Longleftrightarrow\quad
\Bis{\q{testDec$_0$} 0_0^n 0_1v\sigma w\tau x}
\text{\ and\ }
\Bis{\q{\sta} v\sigma w\tau x}
&&(\text{rule for $\q{decOk$_0$\sta}$})\\
& \quad\Longleftrightarrow\quad
\val(v)+1=\val(w)
\text{\ and\ }
\Bis{\q{\sta} v\sigma w\tau x}
&&(\text{Lemma \ref{lem-check-counters}})\\
\end{align*}

For part (c) we have:

\begin{align*}
\Bis{\q{zero$_{0}$ \sta} \sigma w\tau x}
& \quad\Longleftrightarrow\quad
\Bis{\q{decOk$_0$} 0_0^n\sigma w\tau x}
&&(\text{rule for $\q{zero$_0$\sta}$ and part (b)})\\
& \quad\Longleftrightarrow\quad
\val(0_0^n)+1=\val(w)
\text{\ and\ }
\Bis{\q{\sta} v\sigma w\tau x}
&&(\text{Lemma \ref{lem-check-counters}})\\
& \quad\Longleftrightarrow\quad
\val(w)=1
\text{\ and\ }
\Bis{\q{\sta} v\sigma w\tau x}\text{, where}
&&\\
&\quad\phantom{\Longleftrightarrow}\quad\ \ v\text{ is the $0$-counter with $\val(v)=0$}&&
\end{align*}

For part (d) we have:

\begin{align*}
\Bis{\q{dec$_{0}$ \sta} \sigma w\tau x}
& \quad\Longleftrightarrow\quad
\exists v\in\Omega_0^n:\Bis{\q{decOk$_0$ \sta} v\sigma w\tau x}
&&(\text{rules for $\q{dec$_0$ \sta}$ and $\q{dec$_0^{(i)}$ \sta}$ })\\
& \quad\Longleftrightarrow\quad
\exists v\in\Omega_0^n:\val(v)+1=\val(w)
\text{\ and\ }
\Bis{\q{\sta} v\sigma w\tau x}
&&(\text{part (b)})\\
& \quad\Longleftrightarrow\quad
\val(w)\not=0\text{\ and\ }
\Bis{\q{\sta} v\sigma w\tau x}\text{, where $v$}\\
&
\quad\phantom{\Longleftrightarrow}\quad\ \   \text{is the $\ell$-counter with $\val(v)+1=\val(w)$}
&&
\end{align*}
\noindent
{\em Induction step.}
In the following let $\ell \in[0,k-1]$, let $x\in\Omega^*$, let $\sigma,\tau\in\Omega_{\ell+2}$ and let
$v,w$ be $(\ell+1)$-counters.
Let $m \defeq \Tow(n,\ell+1)-1$.
We write~$c_i$ for the $\ell$-counter with $\val(c_i) = i$ for each $i\in[0,m]$.

For part (a) we obtain the following equivalences:
\begin{align*}
\Bis{\q{ones$_{\ell+1}$ \sta} x}
& \quad\Longleftrightarrow\quad \Bis{\q{ones$_{\ell}$ ones$_{\ell+1}^1$ \sta} 1_{\ell+1} x} && (\text{rule for $\q{ones$_{\ell+1}$ \sta}$}) \\
& \quad\Longleftrightarrow\quad \Bis{\q{ones$_{\ell+1}^1$ \sta} c_m 1_{\ell+1} x} && (\text{ind.\ hyp.\ on (a)}) \\
& \quad\Longleftrightarrow\quad \Bis{\q{dec$_{\ell}$  ones$_{\ell+1}^1$ \sta} 1_{\ell+1} c_m 1_{\ell+1} x}  && (\text{rule for $\q{ones$_{\ell+1}^1$ \sta}$}) \\
& \qquad \quad\quad \text{ or } \Bis{\q{zero$_{\ell}$ \sta} 1_{\ell+1} c_m 1_{\ell+1} x} \\
&\quad\Longleftrightarrow\quad \Bis{\q{ones$_{\ell+1}^1$ \sta} c_{m-1} 1_{\ell+1} c_m 1_{\ell+1} x} &&
(\text{ind.\ hyp.\ on (d),(c),$m>1$}) \\
& \quad\Longleftrightarrow\quad \cdots \\
& \quad\Longleftrightarrow\quad \Bis{\q{dec$_{\ell}$  ones$_{\ell+1}^1$ \sta} 1_{\ell+1} c_1 \cdots 1_{\ell+1} c_m 1_{\ell+1} x} \\
& \qquad \quad\quad \text{ or } \Bis{\q{zero$_{\ell}$ \sta} 1_{\ell+1} c_1 \cdots 1_{\ell+1} c_m 1_{\ell+1} x} \\
& \quad\Longleftrightarrow\quad \Bis{\q{ones$_{\ell+1}^1$ \sta} c_0 1_{\ell+1} \cdots c_m 1_{\ell+1} x} && (\text{ind.\ hyp.\ on (d)}) \\
& \qquad \quad\quad \text{ or } \Bis{\q{\sta} \underbrace{c_0 1_{\ell+1} \cdots 1_{\ell+1} c_m 1_{\ell+1}}_\text{ones $(\ell+1)$-counter} x}   && (\text{ind.\ hyp.\ on (c)})
\end{align*}
Further we have
\begin{align*}
                      & \Bis{\q{ones$_{\ell+1}^1$ \sta} c_0 1_{\ell+1} \cdots c_m 1_{\ell+1} x} \\
 \quad\Longleftrightarrow\quad & \Bis{\q{dec$_{\ell}$  ones$_{\ell+1}^1$ \sta} 1_{\ell+1} c_0 1_{\ell+1} \cdots c_m 1_{\ell+1} x}  && (\text{rule for $\q{ones$_{\ell+1}^1$ \sta}$}) \\
                      & \text{ or } \Bis{\q{zero$_{\ell}$ \sta} 1_{\ell+1} c_0 1_{\ell+1} \cdots c_m 1_{\ell+1} x} \\
 \quad\Longleftrightarrow\quad & \mathit{false} && (\text{ind.\ hyp.\ on (d),(c)})
\end{align*}
Combining this with the equivalences above yields part~(a) for $\ell+1$.

For part (b) we obtain the following equivalences:
\begin{align*}
\Bis{\q{decOk$_{\ell+1}$ \sta} v\sigma w\tau x}
& \quad\Longleftrightarrow\quad \Bis{\q{\sta} v\sigma w\tau x}&&\\
&\quad\phantom{\Longleftrightarrow}\quad\text{ and }
\Bis{\q{ones$_{\ell+1}$testDec$_{\ell+1}$}0_{\ell+2}v\sigma w\tau x}
 && (\text{rule for $\q{decOk$_{\ell+1}$ \sta}$}) \\
& \quad\Longleftrightarrow\quad
 \Bis{\q{\sta}v\sigma w\tau x} \text{ and }
\Bis{\q{testDec$_{\ell+1}$}w'0_{\ell+2}v\sigma w\tau x}&&\\
&\quad\phantom{\Longleftrightarrow}\quad\ \
\text{where $w'$ is the ones $(\ell+1)$-counter}
&& (\text{part (a)})\\
&\quad\Longleftrightarrow\quad
\Bis{\q{\sta}v\sigma w\tau x}
\text{ and } \val(w)=\val(v)+1
 && (\text{Lemma \ref{lem-check-counters}}) \\
\end{align*}

For part (c) we obtain the following equivalences:
\begin{align*}
& \Bis{\q{zero$_{\ell+1}$\sta}\sigma w\tau x} \\
&\quad\Longleftrightarrow\quad
\Bis{\q{ones$_\ell$zero$_{\ell+1}^1$ \sta}0_{\ell+1}\sigma w\tau x}
 && (\text{rule for }\q{zero$_{\ell+1}$ \sta})\\
&\quad\Longleftrightarrow\quad
\Bis{\q{zero$_{\ell+1}^1$ \sta}c_m0_{\ell+1}\sigma w\tau x}
 && (\text{ind. hyp. on (a)})\\
&\quad\Longleftrightarrow\quad
\Bis{\q{dec$_\ell$zero$_{\ell+1}^1$ \sta}0_{\ell+1} c_m 0_{\ell+1}\sigma w\tau x}&&\\
&\quad\phantom{\Longleftrightarrow}\quad\ \text{ or }
\Bis{\q{zero$_\ell$decOk$_{\ell+1}$ \sta}0_{\ell+1}c_m0_{\ell+1}\sigma w\tau x}
 && (\text{rule for $\q{zero$_{\ell+1}^1$ \sta}$})\\
&\quad\Longleftrightarrow\quad
\Bis{\q{zero$_{\ell+1}^1$ \sta}c_{m-1}0_{\ell+1}c_m0_{\ell+1}\sigma w\tau x}
 && (\text{ind. hyp. on (d),(c) and $m>1$)})\\
&\quad\Longleftrightarrow\quad\cdots &&\\
&\quad\Longleftrightarrow\quad
\Bis{\q{dec$_\ell$zero$_{\ell+1}^1$ \sta}0_{\ell+1}c_1\cdots 0_{\ell+1}c_m0_{\ell+1}\sigma w \tau x}
&&\\
&\quad\phantom{\Longleftrightarrow}\quad\ \text{ or }
\Bis{\q{zero$_\ell$decOk$_{\ell+1}$ \sta}0_{\ell+1}c_1\cdots 0_{\ell+1}c_m0_{\ell+1}\sigma w\tau x}\\
&\quad\Longleftrightarrow\quad
\Bis{\q{zero$_{\ell+1}^1$ \sta}c_00_{\ell+1}c_1\cdots 0_{\ell+1}c_m0_{\ell+1}\sigma w \tau x}
&&\\
&\quad\phantom{\Longleftrightarrow}\quad\ \text{ or }
\Bis{\q{decOk$_{\ell+1}$ \sta}c_00_{\ell+1}c_1\cdots 0_{\ell+1}c_m0_{\ell+1}\sigma w \tau x}
 && (\text{ind. hyp. on (d),(c)})\\
&\quad\Longleftrightarrow\quad
\Bis{\q{zero$_{\ell+1}^1$ \sta}c_00_{\ell+1}c_1\cdots 0_{\ell+1}c_m0_{\ell+1}\sigma w \tau x}
&&\\
&\quad\phantom{\Longleftrightarrow}\quad\ \text{ or }
(\Bis{\q{\sta}\underbrace{c_00_{\ell+1}c_1\cdots 0_{\ell+1}c_m0_{\ell+1}}_{\text{zero $(l+1)$-counter}}\sigma w \tau x}\ \text{and}\ \val(w)=1
 && (\text{part (b)})
\end{align*}
Further we have
\begin{align*}
 & \Bis{\q{zero$_{\ell+1}^1$ \sta}c_00_{\ell+1}c_1\cdots 0_{\ell+1}c_m0_{\ell+1}\sigma w \tau x}&&\\
 \quad\Longleftrightarrow\quad &
\Bis{\q{dec$_{\ell}$  zero$_{\ell+1}^1$ \sta}0_{\ell+1}c_00_{\ell+1}c_1\cdots 0_{\ell+1}c_m0_{\ell+1}\sigma w \tau x}&&\\
&
\text{ or }
\Bis{\q{zero$_\ell$decOk$_{\ell+1}$ \sta}0_{\ell+1}c_00_{\ell+1}c_1\cdots 0_{\ell+1}c_m0_{\ell+1}\sigma w \tau x}
&& (\text{rule for $\q{zero$_{\ell+1}^1$ \sta}$})\\
\quad\Longleftrightarrow\quad& \mathit{false}
&&\text{(ind. hyp. on (d),(c))}
\end{align*}
Combining this with the equivalences above yields part~(c) for $\ell+1$.

For part (d) we obtain the following equivalences:
\begin{align*}
\Bis{\q{dec$_{\ell+1}$\sta}\sigma w\tau x}
&\quad\Longleftrightarrow\quad \exists\sigma_m\in\Omega_{\ell+1}:
\Bis{\q{ones$_{\ell}$dec$_{\ell+1}^1$ \sta}\sigma_m\sigma w\tau x}
&&\text{(rule for $\q{dec$_{\ell+1}$ \sta}$)}\\
&\quad\Longleftrightarrow\quad \exists\sigma_m\in\Omega_{\ell+1}:
\Bis{\q{dec$_{\ell+1}^1$ \sta}c_m\sigma_m\sigma w\tau x}
&&\text{(ind. hyp. on (a))}\\
&\quad\Longleftrightarrow\quad \exists\sigma_{m-1},\sigma_m\in\Omega_{\ell+1}:&&\\
&\quad\phantom{\Longleftrightarrow}\quad\quad
\Bis{\q{dec$_\ell$dec$_{\ell+1}^1$ \sta}\sigma_{m-1}c_m\sigma_m\sigma w\tau x}\\
&\quad\phantom{\Longleftrightarrow}\quad\quad
\text{ or }\Bis{\q{zero$_{\ell}$decOk$_{\ell+1}$ \sta}\sigma_{m-1}c_m\sigma_m\sigma w\tau x}
&&\text{(rule for $\q{dec$_{\ell+1}^1$ \sta}$)}\\
&\quad\Longleftrightarrow\quad \exists\sigma_{m-1},\sigma_m\in\Omega_{\ell+1}:&&\\
&\quad\phantom{\Longleftrightarrow}\quad\quad
\Bis{\q{dec$_{\ell+1}^1$ \sta}c_{m-1}\sigma_{m-1}c_m\sigma_m\sigma w\tau x}
&&\text{(ind. hyp. on (d),(c),$m>1$)}\\
&\quad\Longleftrightarrow\quad\cdots &&\\
&\quad\Longleftrightarrow\quad \exists\sigma_{1},\ldots,\sigma_m\in\Omega_{\ell+1}:&&\\
&\quad\phantom{\Longleftrightarrow}\quad\quad
\Bis{\q{dec$_{\ell+1}^1$ \sta}c_{1}\sigma_{1}\cdots c_m\sigma_m\sigma w\tau x}&&\\
&\quad\Longleftrightarrow\quad \exists\sigma_{0},\ldots,\sigma_m\in\Omega_{\ell+1}:&&\\
&\quad\phantom{\Longleftrightarrow}\quad\quad
\Bis{\q{dec$_\ell$dec$_{\ell+1}^1$ \sta}\sigma_0c_1\sigma_1\cdots c_m\sigma_m\sigma w\tau x}\\
&\quad\phantom{\Longleftrightarrow}\quad\quad
\text{ or }
\Bis{\q{zero$_\ell$decOk$_{\ell+1}$ \sta}\sigma_0c_1\sigma_1\cdots c_m\sigma_m\sigma w\tau x}
&&\text{(rule for $\q{dec$_{\ell+1}^1$}$)}\\
&\quad\Longleftrightarrow\quad \exists\sigma_{0},\ldots,\sigma_m\in\Omega_{\ell+1}:&&\\
&\quad\phantom{\Longleftrightarrow}\quad\quad
\Bis{\q{dec$_\ell$dec$_{\ell+1}^1$ \sta}\sigma_0c_1\sigma_1\cdots c_m\sigma_m\sigma w\tau x}&&\\
&\quad\phantom{\Longleftrightarrow}\quad\quad
\text{ or }
\Bis{\q{decOk$_{\ell+1}$ \sta}c_0\sigma_0c_1\sigma_1\cdots c_m\sigma_m\sigma w\tau x}
&&\text{(ind. hyp. on (c)}\\
&\quad\Longleftrightarrow\quad \exists\sigma_{0},\ldots,\sigma_m\in\Omega_{\ell+1}:&&\\
&\quad\phantom{\Longleftrightarrow}\quad\quad
\Bis{\q{dec$_\ell$dec$_{\ell+1}^1$ \sta}\sigma_0c_1\sigma_1\cdots c_m\sigma_m\sigma w\tau x}&&\\
&\quad\phantom{\Longleftrightarrow}\quad\quad
\text{ or }
(\Bis{\q{\sta}c_0\sigma_0c_1\sigma_1\cdots c_m\sigma_m\sigma w\tau x}&&\\
&\quad\phantom{\Longleftrightarrow}\quad\quad\qquad
\text{ and }\val(w)=\val(c_0\sigma_0\cdots c_m\sigma_m)+1) &&\text{$(\star)$}\\
&&&\text{(ind. hyp. on (b))}
\end{align*}
Further we have for all $\sigma_0,\ldots,\sigma_m\in\Omega_{\ell+1}$:
\begin{align*}
 & \Bis{\q{dec$_{\ell}$dec$_{\ell+1}^1$ \sta}
\sigma_0c_1\sigma_1\cdots c_m\sigma_m\sigma w\tau x}&\\
\Longleftrightarrow\quad &
\Bis{\q{dec$_{\ell+1}^1$ \sta}c_0\sigma_0c_1\sigma_1\cdots c_m\sigma_m\sigma w\tau x}
&\text{(ind. hyp. on (d))}\\
\Longleftrightarrow\quad &
\exists\sigma_{-1}\in\Omega_{\ell+1}:
\Bis{\q{dec$_\ell$dec$_{\ell+1}^1$\sta}\sigma_{-1}c_0\sigma_0c_1\sigma_1\cdots c_m\sigma_m\sigma w\tau x}&\\
\phantom{\Longleftrightarrow}\quad &
\qquad\qquad\qquad\text{or}\ \Bis{\q{zero$_\ell$decOk$_{\ell+1}$ \sta}
\sigma_{-1}c_0\sigma_0c_1\sigma_1\cdots c_m\sigma_m\sigma w\tau x}\\
\Longleftrightarrow\quad &
\mathit{false}&\text{(ind. hyp. on (d),(c))}
\end{align*}
Hence $(\star)$ is equivalent to
$$
\exists\sigma_0,\ldots,\sigma_m\in\Omega_{\ell+1}:
\Bis{\q{\sta}c_0\sigma_0c_1\sigma_1\cdots c_m\sigma_m\sigma w\tau x}\
\text{ and }\ \val(w)=\val(c_0\sigma_0\cdots c_m\sigma_m)+1
$$
which is in turn equivalent to
      $$\val(w) \ne 0\ \text{ and } \Bis{\q{\sta} v \sigma w \tau x},\ \text{where $v$ is the $\ell$-counter with $\val(v) + 1 = \val(w)$},$$
which shows part~(d) for $\ell+1$.
\qed
\end{proof}

\subsection{Simulating a Transducer Machine}{\label{sec-simulating}}

Returning to our overall reduction, let us recall we have fixed
 a deterministically
terminating transducer machine
$\T=(\Tow(k,n),T_1,T_2)$ with respect to which
$0^{\Tow(k,n)}$ is a dead end, i.e. we have the following in total:
\begin{itemize}
\item $T_1,T_2:\{0,1\}^*\rightarrow\Upsilon^*$ are letter-to-letter transducers,
\item for each $z\in\{0,1\}^{\Tow(k,n)}$ there is at most one $z'$ with
$T_1(z)=T_2(z')$,
\item assume $z_0,\ldots,z_t\in\{0,1\}^{\Tow(k,n)}$ such that
$z_0=1^{\Tow(k,n)}$ and $T_1(z_i)=T_2(z_{i+1})$ for each $i\in[0,t-1]$,
and $z_t$ is a dead end (we defined $\Last(\T) \defeq z_t$), and
\item $0^{\Tow(k,n)}$ is a dead end with respect to $\T$.
\end{itemize}

We include rules so that $\Bis{\q{start}}$ holds if and only if
 $\Last(\T)=0^{\Tow(k,n)}$,
 thus completing our reduction.
The PDA~$\P$ will be able to push $z_0, z_1, \ldots$ on the stack,
 where each word~$z_i$ is encoded as a $k$-counter, say~$d_i$, in the obvious way:
 $z_i = \eta(d_i)$ where $\eta: \Omega^* \to \Omega_k^*$ denotes the homomorphism with
  $\eta(\sigma) = \sigma$ for $\sigma \in \Omega_k$ and
  $\eta(\sigma) = \varepsilon$ otherwise.
We emphasize that in comparison to the counters that were present in the proof
of Lemma \ref{lem-build-counters} for each $i\in[0,t]$ we do not generally have $\val(d_i)=i$: Instead
we have $z_i=\eta(d_i)$, in particular the sequence $z_0,\ldots,z_t$ and thus the sequence
$d_0,\ldots,d_t$ is determined.
The $d_i$ will be separated on the stack by the symbol $\$ \defeq 0_{k+1}$.
We include rules such that the following holds:

\begin{lemma} \label{lem-high-level}
Let $x \in \Omega^*$ and let $w_1, w_2, w_3$ be $k$-counters.
\begin{itemize}
 \item[(a)]
  Then $\Bis{\q{testFin} w_1 \$ x}$ iff $\eta(w_1)=0^{\Tow(k,n)}$.
 \item[(b)]
  Then $\Bis{\q{testTran} w_3 \$ w_2 \$ w_1 \$ x}$
   iff $T_1(\eta(w_1)) = T_2(\eta(w_2))$.
 \item[(c)]
  Then $\Bis{\q{start}}$
   iff  $z_t=0^{\Tow(k,n)}$.
\end{itemize}
\end{lemma}

For Lemma~\ref{lem-high-level} we include the following rules:

\begin{align*}
 \q{start} & \ctran{} \q{ones$_k$ fin} \$ && \text{push $\$$ and the encoded $z_0$} \\
 \q{fin}   & \ctran{\Def} \{ \q{testFin} , \ \q{next} \$ \} && \text{test $0^{\Tow(k,n)}$ OR do the next~$z_i$} \\
 \ql{testFin} \left( \Omega_{\le k-1}^* \Omega_k \right)^* \$ & \btran{\{1_k\}\mapsto b\ \shuf\
(\Omega \setminus \{1_k\}) \mapsto a} \ql{popAll}
    && \text{rule {\`a} la Lemma~\ref{lem-transducers}~(a) to test $0^{\Tow(k,n)}$} \\
 \qr{testFin} \left( \Omega_{\le k-1}^* \Omega_k \right)^* \$ & \btran{\Omega \mapsto a} \qr{popAll}
&&\text{rule {\`a} la Lemma~\ref{lem-transducers}~(a) to test $0^{\Tow(k,n)}$}\\
 \ql{popAll} \omega & \btran{a} \ql{popAll} &&\text{for all $\omega\in\Omega$: Erase stack content}\\
 \qr{popAll} \omega & \btran{a} \qr{popAll} &&\text{for all $\omega\in\Omega$: Erase stack content}\\
 \q{next}     & \ctran{\Def} \left\{ \q{ones$_{k-1}$ next$^1$} \sigma \mid \sigma \in \Omega_k \right\}  && \text{choose the next~$z_i$} \\
 \q{next$^1$} & \ctran{\Def} \left\{ \q{dec$_{k-1}$  next$^1$} \sigma , \right.  \\
              & \hspace{10mm}\left.  \q{zero$_{k-1}$  tran} \sigma \mid \sigma \in \Omega_k \right\} \\
 \q{tran}     & \ctran{\Att} \left\{ \q{ones$_k$ testTran} \$, \ \right. && \text{test whether new $z_i$ is ok} \\
              & \hspace{10mm}\left.  \q{fin} \right\}      && \text{OR continue} \\
 \ql{testTran} \Omega_{\le k}^* \$  \Omega_{\le k}^* \$
         & \btran{\Omega \mapsto a} \ql{testTran$^1$} && \text{rule {\`a} la Lemma~\ref{lem-transducers}}\\
 \qr{testTran} \Omega_{\le k}^* \$
         & \btran{\Omega \mapsto aa} \qr{testTran$^1$}&& \text{rule {\`a} la Lemma~\ref{lem-transducers}}\\
 \ql{testTran$^1$} \left( \Omega_{\le k-1}^* \Omega_k \right)^* \$
         & \btran{T_1\ \shuf\ ((\Omega \setminus \Omega_k) \mapsto a)} \ql{stop$_k$} && \text{rule {\`a} la Lemma~\ref{lem-transducers} for $T_1$}\\
 \qr{testTran$^1$} \left( \Omega_{\le k-1}^* \Omega_k \right)^* \$
         & \btran{T_2\ \shuf\ ((\Omega \setminus \Omega_k) \mapsto a)} \qr{stop$_k$} && \text{rule {\`a} la Lemma~\ref{lem-transducers} for $T_2$} \\
\end{align*}

\begin{proof}[of Lemma~\ref{lem-high-level}]
Parts (a) and~(b) are simple consequences of the above rules and Lemma~\ref{lem-transducers} (a) and~(b),
respectively.
For part~(c) we have:
\begin{align*}
                        \Bis{\q{start}}
&  \quad\Longleftrightarrow\quad  \Bis{\q{fin}} d_0 \$ &&                                     \text{rule for $\q{start}$ and Lemma~\ref{lem-build-counters}~(a)} \\
&  \quad\Longleftrightarrow\quad \eta(d_0)=0^{\Tow(k,n)} \text{ or } \Bis{\q{next} \$ d_0 \$} && \text{rule for $\q{fin}$ and part~(a)}
\end{align*}
By the rules for $\q{next}$ and $\q{next$^1$}$ and reasoning as in Lemma~\ref{lem-build-counters} we have
 $$\text{$\Bis{\q{next} \$ d_0 \$}$
 $\quad\Longleftrightarrow\quad$ there is a $k$-counter~$w_1$ with $\Bis{\q{tran} w_1 \$ d_0 \$}$}.$$
By the rules for $\q{tran}$ and part~(b) we have
$$\text{$\Bis{\q{tran} w_1 \$ d_0 \$}$ $\quad\Longleftrightarrow\quad$
$w_1 = d_1$\ and\ $\Bis{\q{fin} w_1 \$ d_0 \$}$}.$$
It follows that:
\begin{align*}
  \Bis{\q{next} \$ d_0 \$}
& \quad\Longleftrightarrow\quad \Bis{\q{fin} d_1 \$ d_0 \$} \\
& \quad\Longleftrightarrow\quad \eta(d_1)=0^{\Tow(k,n)}\ \text{ or }\ \Bis{\q{next} \$ d_1 \$ d_0 \$} && \text{rule for $\q{fin}$ and part~(a)}
\end{align*}
Combining this with the equivalences above we obtain:
\begin{align*}
                        \Bis{\q{start}}
& \quad\Longleftrightarrow\quad  \bigvee_{i=0}^1 \eta(d_i)=0^{\Tow(k,n)}\ \text{ or }\ \Bis{\q{next} \$ d_1 \$ d_0 \$}\\
\end{align*}
By iterating this reasoning, we obtain:
\begin{align*}
                        \Bis{\q{start}}
& \quad\Longleftrightarrow\quad  \bigvee_{i=0}^t \eta(d_i)=0^{\Tow(k,n)}\ \text{ or }\
\Bis{\q{next} \$ d_t \$ d_{t-1} \$ \cdots d_0 \$}\\
& \quad\Longleftrightarrow\quad \bigvee_{i=0}^t z_i=0^{\Tow(k,n)}\ \text{ or }\ \Bis{\q{next} \$ d_t \$ d_{t-1} \$ \cdots d_0 \$}\\
& \quad\Longleftrightarrow\quad z_t=0^{\Tow(k,n)}\ \text{ or }\ \Bis{\q{next} \$ d_t \$ d_{t-1} \$ \cdots d_0 \$},
\end{align*}
where the last equivalence follows from the fact that $0^{\Tow(k,n)}$ is a dead end.
By reasoning as above we have:
\begin{align*}
  \Bis{\q{next} \$ d_t \$ \cdots d_0 \$}
& \quad\Longrightarrow\quad \text{there is a $k$-counter $w$ with } \Bis{\q{tran} w \$ d_t \$ \cdots d_0 \$} \\
& \quad\Longrightarrow\quad T_1(\eta(d_t)) = T_2(\eta(w)) \;,
\end{align*}
which is false, by definition of~$t$.
We conclude that $\Bis{\q{start}}$ holds iff $z_t=0^{\Tow(k,n)}$.
\qed
\end{proof}

By Lemma~\ref{lem-high-level}~(c) we have completed the reduction, and hence the proof of Theorem~\ref{thm-main}.

\subsection{Normedness} \label{sub-normedness}

We can strengthen Theorem~\ref{thm-main}.
Given a PDA $\P$ with control state set $Q$, we say a state $s$ of $\mathcal{S}(\P)$ is {\em normed} if
every state $t$ that is reachable from $s$ can reach some deadlock $t'$, where
$t'=q$ for some $q\in Q$ (i.e. the stack is empty).

\begin{theorem} \label{thm-normedness}
 PDA bisimilarity for is nonelementary, even when the initial states are normed.
\end{theorem}
\begin{proof}
Recall that \eqref{eq-stop} requires for $x \in \Omega^*$ and $\sigma \in \Omega_{\ell+1}$ and $\ell$-counters~$w$ that we have
\begin{equation*}
 \ql{stop$_\ell$} x \sim \qr{stop$_\ell$} w \sigma x \;.
\end{equation*}
We satisfied this by not giving any rules for $\ql{stop$_\ell$}$ and $\qr{stop$_\ell$}$, thus creating deadlocks.
Thus we allowed to reach states that are not normed.
We show that the construction can be amended to avoid the latter.
Assume that $x = w' \sigma' x'$ holds for an $\ell$-counter~$w'$, $\sigma' \in \Omega_{\ell+1}$ and $x'\in\Omega^*$.
Add the following rules:
\begin{align*}
 \ql{stop$_\ell$} \Omega_{\le \ell}^* \Omega_{\ell+1}
         & \btran{\Omega \mapsto aa} \ql{popAll} \\
 \qr{stop$_\ell$} \Omega_{\le \ell}^* \Omega_{\ell+1}  \Omega_{\le \ell}^* \Omega_{\ell+1}
         & \btran{\Omega \mapsto a} \qr{popAll}
\end{align*}
and recall that rules for $\ql{popAll}$ and $\qr{popAll}$ were given previously.
Intuitively, $w'\sigma'$ is popped off the left stack with ``half'' speed, while $w \sigma w' \sigma'$ is popped off the right stack with ``full'' speed.
Afterwards the stacks are of equal height and can be fully erased.
Now \eqref{eq-stop} is satisfied, and these states are normed.
For the latter fact, observe all reachable stacks are of the form
$
 (( \ldots (( \Omega_0^* \Omega_1)^* \Omega_2)^* \ldots )^*\$)^*
$
and $\ql{popAll}$ or $\qr{popAll}$ is always reachable.
The only remaining problem is to guarantee that, whenever we require $\ql{stop$_\ell$} x \sim \qr{stop$_\ell$} w \sigma x$,
$x$ is indeed prefixed by $w' \sigma'$. Note that $\mathsf{stop}_k$ is introduced in the rule for $\mathsf{testTran}$
and other occurrences of $\mathsf{stop}_\ell$ are used in the counter-related rules based on $\mathsf{testDec}_\ell$.
Consequently, if the above requirement is not already satisfied, $w$ must be the ones $\ell$-counter.
To prevent this from happening, for $\mathsf{testTran}$ we shall add an extra ones $k$-counter at the beginning
of the simulations. For $\mathsf{testDec}$, we shall eliminate the need
for decrement tests involving ones counters
by introducing a new symbol $\mathsf{zOnes}$ that will correspond to the predecessor of a ones counter.

This can be achieved by the following modifications:
 \begin{itemize}
   \item Replace $\q{start} \ctran{} \q{ones$_k$ fin} \$$ \ with $\q{start} \ctran{} \q{ones$_k$ start$^0$} \$$ and $\q{start$^0$} \ctran{} \q{ones$_k$ fin} \$$;
    i.e., push one additional ones $k$-counter and~$\$$ in the beginning.
   \item Extend~$B$ by symbols $\textsf{zOnes$_\ell$}$ with $\ell \in [0,k+2]$, whose role is to push an $\ell$-counter~$w$ with
   $\val(w) =\Tow(\ell+1,n)-2$. (This value is one less than the value of a ones $\ell$-counter.)
   \item Replace $\q{ones$_\ell$ \sta}   \ctran{}     \q{ones$_{\ell-1}$  ones$_\ell^1$ \sta} 1_\ell$ with
   \begin{align*}
      & \q{ones$_\ell$ \sta}   \ctran{}     \q{ones$_{\ell-1}$  ones$^0_\ell$ \sta} 1_\ell \quad \text{and add} \\
      & \q{ones$^0_\ell$ \sta}   \ctran{}     \q{zOnes$_{\ell-1}$  ones$_\ell^1$ \sta} 1_\ell\,.
\intertext{Similarly, replace $\q{dec$_\ell$ \sta} \ctran{\Def} \left\{ \q{ones$_{\ell-1}$  dec$_\ell^1$ \sta} \sigma \mid \sigma \in \Omega_\ell \right\}$ with}
      & \q{dec$_\ell$ \sta} \ctran{\Def} \left\{ \q{ones$_{\ell-1}$  dec$_\ell^0$ \sta} \sigma \mid \sigma \in \Omega_\ell \right\} \quad \text{and add} \\
      & \q{dec$^0_\ell$ \sta}  \ctran{\Def} \left\{ \q{zOnes$_{\ell-1}$  dec$_\ell^1$ \sta} \sigma \mid \sigma \in \Omega_\ell \right\}\,.
\intertext{Similarly, replace $\q{zero$_\ell$ \sta} \ctran{}     \q{ones$_{\ell-1}$  zero$_\ell^1$ \sta} 0_\ell$ with}
      & \q{zero$_\ell$ \sta}  \ctran{}     \q{ones$_{\ell-1}$  zero$_\ell^0$ \sta} 0_\ell \quad \text{and add} \\
      & \q{zero$^0_\ell$ \sta} \ctran{}     \q{zOnes$_{\ell-1}$  zero$_\ell^1$ \sta} 0_\ell\,.
   \end{align*}
   \item
    For $\ell \ge 1$ add the following rules:
    \begin{align*}
      & \q{zOnes$_\ell$ \sta} \ctran{}     \q{ones$_{\ell-1}$  zOnes$^0_\ell$ \sta} 1_\ell \\
      & \q{zOnes$^0_\ell$ \sta} \ctran{}     \q{zOnes$_{\ell-1}$  zOnes$^1_\ell$ \sta} 1_\ell \\
      & \q{zOnes$_\ell^1$ \sta} \ctran{\Def} \left\{ \q{dec$_{\ell-1}$ zOnes$_\ell^1$ \sta} 1_\ell ,
                                                   \q{zero$_{\ell-1}$ \sta} 0_\ell \right\}
    \end{align*}
    Finally, add
    $\q{zOnes$_0$ \sta} \ctran{}  \q{\sta} 0_0 1_0^{n-1}$.
 \end{itemize}
By these modifications, whenever we require $\ql{stop$_\ell$} x \sim \qr{stop$_\ell$} w \sigma x$ in our proofs,
 the word $x$ will start with an $\ell$-counter and an $\Omega_{\ell+1}$-symbol.
\qed
\end{proof}

\bibliographystyle{plain} 
\bibliography{db}

%
\end{document}